\documentclass[smallextended]{svjour3-arxiv}     

\usepackage{amsmath}

\usepackage{amsthm}

\smartqed  
\usepackage{graphicx}
\usepackage{wrapfig}
\usepackage{url}
\newcommand\eat[1]{}
\usepackage{tikz}
\usepackage{booktabs}  
\usetikzlibrary{arrows}
\usepackage{tablefootnote}
\usepackage{slashbox}
	
        \journalname{Working Paper}

\usepackage{array,xspace,multirow,hhline,graphicx,xcolor,tikz,colortbl,tabularx,amsmath,amssymb,amsfonts}
\usepackage[round]{natbib}
\usetikzlibrary{shapes}
\usepackage{algorithm,algorithmic}
\usepackage{mathrsfs}
\usepackage{enumitem}
\setenumerate[1]{label=(\emph{\roman{*}}),ref=(\emph{\roman{*}}),leftmargin=*}

		\newcommand{\bartonnew}[1]{\textcolor{black}{#1}}

\usepackage{txfonts}
\usepackage{eucal} 

\newcommand{\pref}{\succsim\xspace}

\usepackage{tikz}

\newcommand{\midd}{\mathbin{: }}

\usepackage{mathtools}
\DeclarePairedDelimiter\ceil{\lceil}{\rceil}
\DeclarePairedDelimiter\floor{\lfloor}{\rfloor}

	\newcommand{\Pref}[1][]{
		\ifthenelse{\equal{#1}{}}{\mathrel \succsim}{\mathop{R_{#1}}}
	}    
				\newcommand{\spref}{\ensuremath{\succ}}                                      
	\newcommand{\sPref}[1][]{                  
		\ifthenelse{\equal{#1}{}}{\mathrel \succ}{\mathop{P_{#1}}}
	}                                          
	\newcommand{\Indiff}[1][]{                 
		\ifthenelse{\equal{#1}{}}{\mathrel \sim}{\mathop{\sim_{#1}}}
	}
	\newcommand{\prefset}[1][]{\ifthenelse{\equal{#1}{}}{\mathcal{R}}{\mathcal{R}_{#1}}}
	
	\newcommand{\indiff}{\mathbin \sim\xspace}

\newcommand{\newrule}{{EAR}\xspace}

\newcommand{\harisnew}[1]{\textcolor{black}{#1}}

\let\enumtemp=\enumerate
\def\enumerate{\enumtemp\itemsep 1pt}
\let\itemtemp=\itemize
\def\itemize{\itemtemp\itemsep 1pt}

\newcommand{\Omit}[1]{}

\begin{document}


\title{The Expanding Approvals Rule: \\
Improving Proportional Representation and Monotonicity}

	\author{Haris Aziz \and Barton E. Lee}

	\institute{%
	  H. Aziz \and 
	 B. E. Lee \at
	 Data61, CSIRO and UNSW,
	 	  Sydney 2052 , Australia \\
	 	  Tel.: +61-2-8306\,0490 \\
	 	  Fax: +61-2-8306\,0405 \\
	 \email{haris.aziz@unsw.edu.au, barton.e.lee@gmail.com}
}

	\newlength{\wordlength}
	\newcommand{\wordbox}[3][c]{\settowidth{\wordlength}{#3}\makebox[\wordlength][#1]{#2}}
	\newcommand{\mathwordbox}[3][c]{\settowidth{\wordlength}{$#3$}\makebox[\wordlength][#1]{$#2$}}
    		\renewcommand{\algorithmicrequire}{\wordbox[l]{\textbf{Input}:}{\textbf{Output}:}} 
    		 \renewcommand{\algorithmicensure}{\wordbox[l]{\textbf{Output}:}{\textbf{Output}:}}

\date{}

	%

\maketitle

\begin{abstract}
Proportional representation (PR) is often discussed in voting settings as a major desideratum. For the past century or so, it is common both in practice and in the academic literature to jump to single transferable vote (STV) as the solution for achieving PR. 
Some of the most prominent electoral reform movements around the globe are pushing for the adoption of STV.
It has been termed a major open problem to design a voting rule that satisfies the same PR properties as STV and better monotonicity properties. 
In this paper, we first present a taxonomy of proportional representation axioms for general weak order preferences, some of which generalise and strengthen previously introduced concepts. 
We then present a rule called Expanding Approvals Rule (EAR) that satisfies properties stronger than the central PR axiom satisfied by STV,  can handle indifferences in a convenient and computationally efficient manner, and also satisfies better candidate monotonicity properties. 
In view of this, our proposed rule seems to be a compelling solution for achieving proportional representation in voting settings. 
\end{abstract}



\section{Introduction}




\begin{quote}
Of all modes in which a national representation can possibly be constituted, this one [STV] affords the best security for the intellectual qualifications desirable in the representatives---John Stuart Mill (Considerations on Representative Government, 1861).
	\end{quote}

\begin{quote}
	A major unsolved problem is whether there exist rules that retain the important political features of STV and are also more monotonic---\citet{Wood97a}.
	\end{quote}
	
	We consider a well-studied voting setting in which $n$ voters express ordinal preferences over $m$ candidates and based on the preferences $k\leq m$ candidates are selected. The candidates may or may not be from particular parties but voters express preferences directly over individual candidates.\footnote{The setting is referred to as a preferential voting system. It is more general and flexible than settings in which voters vote for their respective parties and then the number of seats apportioned to the parties is proportional to the number of votes received by the party~\citep{Puke14a}. } 
This kind of voting problem is not only encountered in parliamentary elections but to form any kind of representative body. When making such a selection by a voting rule, a desirable requirement is that of proportional representation. Proportional representation stipulates that voters should get representation in a committee or parliament according to the strengths of their numbers. It is widely accepted that proportional representation is the fairest way to reflect the diversity of opinions among the voters.\footnote{Proportional representation may be the fairest way for representation but it also allows for extreme group to have some representation at least when the group is large enough. PR also need not be the most effective approach to a stable government. \citet{Blac58a} wrote that ``It [PR] makes it difficult to form a cabinet which can command a parliamentary majority and so makes for weak government.''}

 For the last 120 years or so, the most widely used and accepted way to achieve it is via single transferable vote (STV)~\citep{Blac58a,TiRi00a} and its several variants. In fact STV is used for elections in several countries including Australia, Ireland, India, and Pakistan. 
 It is also used to select representative committees in hundreds of settings including professional organisations, scientific organizations, political parties, school groups, and university student councils all over the globe.\footnote{Notable uses of STV include Oscar nominations, internal elections of the British Liberal Democrats, and selection of Oxford Union, Cambridge Union, and Harvard/Radcliffe Undergraduate Councils.} 

The reason for the widespread adoption of STV is partly due to the fact that it has been promoted to satisfy  proportional representation axioms. In particular, STV satisfies a key PR axiom called \emph{Proportionality for Solid Coalitions} (PSC)~\citep{TiRi00a,Wood94a}. \citet{Tide95a} argues that ``\emph{It is the fact that STV satisfies PSC that justifies describing STV as a system of proportional representation.}'' \citet{Wood97a} also calls the property the ``\emph{essential feature of STV, which makes it a system of proportional representation}.''
 \citet{Dumm84a} motivated PSC as a minority not requiring to coordinate its report and that it should deserve some high preferred candidates to be selected as long as enough voters are `solidly committed to such candidates.' PSC captures the idea that as long as voters have the same top candidates (possibly in different ordering), they do not need to coordinate their preferences to get a justified number of such candidates selected.
In this sense PSC is also similar in spirit to the idea that if a clone of a candidate is introduced, it should not affect the selection of the candidate. 
 Voters from the same party not having to coordinate their reports as to maximize the number of winners from their own party can be viewed as a weak form of group-strategyproofness. PSC can also be seen as a voter's vote not being wasted due to lack of coordination with like-minded voters. 
PSC has been referred to as ``a sine qua non for a fair election rule'' by \citet{Wood94a}.\footnote{There are two PSC axioms that differ in only whether the Hare quota is used or whether the Droop quota is used. The one with respect to the Droop quota has also been referred to as DPC (Droop's proportionality criterion)~\citep{Wood94a}. 
\citet{Wood94a} went as far as saying that ``I assume that no member of the Electoral Reform Society will be satisfied with anything that does not satisfy DPC.'' } 

Although STV is not necessarily the only rule satisfying PR properties, it is at times synonymous with proportional representation in academia and policy circles. The outcome of STV can also be computed efficiently which makes it suitable for large scale elections.\footnote{Although it is easy to compute one outcome of STV, checking whether a certain set is a possible outcome of STV is NP-complete~\citep{CRX09a}.} Another reason for the adoption of STV is historical. Key figures proposed ideas related to STV or pushed for the adoption of STV. The ideas behind STV can be attributed to several thinkers including C. Andrae, T. Hare
, H. R. Droop, and T. W. Hill. For a detailed history of the development of STV family of rules, please see the article by \citet{Tide95a}. In a booklet, \citet{Aiya30a} explains the rationale behind different components of the STV rule. 
STV was supported by influential intellectuals such as 
John Stuart Mill 
who placed STV ``\emph{among the greatest improvements yet made in the theory and practice of government}.'' \citet{BoGr00a} note the British influence on the spread of STV among countries with historical association with Great Britain. 

With historical, normative, and computational motivation behind it, STV has become the  `go to' rule for PR and has strong support.\footnote{One notable exception was philosopher Michael Dummett who was a stringent critic of STV. He proposed a rival PR method called the Quota Borda System (QBS) and pushed its case~\citep{Dumm84a,Dumm97a}. However, even he agreed that in terms of achieving PR, ``[STV] guarantees representation for minorities to the greatest degree to which any possible electoral system is capable of doing''~\citep{Dumm97a}[page 137].} It is also vigorously promoted by prominent electoral reform movements across the globe including the Proportional Representation Society of Australia (\url{http://www.prsa.org.au}) and the Electoral Reform Society (\url{https://www.electoral-reform.org.uk}).

Despite the central position of STV, it is not without some flaws. It is well-understood that it violates basic monotonicity properties even when selecting a single candidate~\citep[see e.g., ][]{DoKr77a,Zwic15a}. Increasing the ranking of the winning candidate may result in the candidate not getting selected. STV is also typically defined for strict preferences which limits its ability to tackle more general weak orders. 
There are several settings where voters may be indifferent between two candidates because  the candidates have the same characteristics that the voter cares about. It could also be that the voter does not have the cognitive power or time to distinguish between two candidates and does not wish to break ties arbitrarily. 
It is not clearly resolved in the literature how STV can be extended to handle weak orders without compromising on its computational efficiency or some of the desirable axiomatic properties it satisfies.\footnote{\citet{Hill01a} and \citet{Meek94a} propose one way to handle indifferences however, this leads to an algorithm that may take time $O(m!)$.}
The backdrop of this paper is that improving upon STV in terms of both PR as well as monotonicity has been posed as a major challenge~\citep{Wood97a}. 




\bigskip

\paragraph{Contributions}




We propose a new voting rule called \emph{Expanding Approvals Rule (\newrule)} that has several advantages.
(1) It satisfies an axiom called Generalised PSC that is stronger than PSC. (2) It satisfies some natural monotonicity criteria that are not satisfied by STV.
(3) It is defined on general weak preferences rather than just for strict preferences and hence constitutes a flexible and general rule that finds a suitable outcome in polynomial time for both strict and dichotomous preferences. Efficient computation of a rule is an important concern when we deal with election of large committees. 

Our work also helps understand the specifications under which different variants of STV satisfy different PR axioms. 
Apart from understanding how far STV and \newrule satisfy PR axioms, one of the conceptual contributions of this paper is to define a taxonomy of PR axioms based on PSC and identify their relations with each other. In particular, we propose a new axiom for weak preferences called \emph{Generalised PSC} that simultaneously generalises PSC (for strict preferences) and proportional justified representation (for dichotomous preferences).



\section{Model and Axioms}

In this section, we lay the groundwork of the paper by first defining the model and then formalizing the central axioms by which proportional representation rules are judged. 

\subsection{Model}

We consider the standard social choice setting with a set of voters $N=\{1,\ldots, n\}$, a set of candidates $C=\{c_1,\ldots, c_m\}$ and a preference profile $\pref=(\pref_1,\ldots,\pref_n)$ such that each $\pref_i$ is a complete and transitive relation over $C$. 
Based on the preference profile, the goal is to select a committee $W\subset C$ of predetermined size $k$.
Since our new rule is defined over weak orders rather than strict orders, we allow the voters to express weak orders.

			 			We write~$a \pref_i b$ to denote that voter~$i$ values candidate~$a$ at least as much as candidate~$b$ and use~$\spref_i$ for the strict part of~$\pref_i$, i.e.,~$a \spref_i b$ iff~$a \pref_i b$ but not~$b \pref_i a$. Finally, $\indiff_i$ denotes~$i$'s indifference relation, i.e., $a \indiff_i b$ if and only if both~$a \pref_i b$ and~$b \pref_i a$.
			 			The relation $\pref_i$ results in (non-empty) equivalence classes $E_i^1,E_i^2, \ldots, E_i^{m_i}$ for some $m_i$ such that $a\spref_i a'$ if and only if $a\in E_i^l$ and $a'\in E_i^{l'}$ for some $l<l'$. Often, we will use these equivalence classes to represent the preference relation of a voter as a preference list
			 			$i\midd E_i^1,E_i^2, \ldots, E_i^{m_i}$. If candidate $c$ is in $E_i^j$, we say it has \emph{rank} $j$ in voter $i$'s preference. 
			 		For example, we will denote the preferences $a\indiff_i b\spref_i c$ by the list $i:\ \{a,b\}, \{c\}$. 
					If each equivalence is of size $1$, the preferences will be called strict preferences or linear orders. Strict preferences will be represented by a comma separated list of candidates. If for each voter, the number of equivalence classes is at most two, the preferences are referred to as dichotomous preferences. When the preferences of the voters are dichotomous, the voters can be seen as approving a subset of voters. In this case for each voter $i\in N$, the first equivalence class $E_i^1$ is also referred to as the approval set $A_i$. The vector $A=(A_1,\ldots, A_n)$ is referred to as the \emph{approval ballot profile}. \harisnew{Since our model concerns ordinal preferences, when a voter $i$ is completely indifferent between all the candidates, it means that $E_i^1=C$ and we do not ascribe any intensity with which these candidates are liked or disliked by voter $i$. None of our results connecting dichotomous preferences with approval-based committee voting are dependent on how complete indifference is interpreted in terms of all approvals versus all disapprovals.}
					
The model allows for voters to express preference lists that do not include some candidates. In that case, the candidates not included in the list will be assumed to form the last equivalence class.

%


\subsection{PR under Strict Preferences}

In order to understand the suitability of voting rules for proportional representation, we recap the central PR axiom from the literature. It was first mentioned and popularised by \citet{Dumm84a}. It is defined for strict preferences.

\begin{definition}[Solid coalition]
A set of voters $N'$ is a \emph{solid coalition} for a set of candidates $C'$ if every voter in $N'$ strictly prefers every candidate in $C'$ ahead of every candidate in $C\backslash C'$. That is, for all $i\in N'$ and for any $c'\in C'$
$$\forall c\in C\backslash C' \quad c'\spref_i c. $$
The candidates in $C'$ are said to be supported by voter set $N'$.
\end{definition}

\bartonnew{Importantly, the definition of a solid coalition does not require voters to maintain the same order of strict preferences among candidates in $C'$ nor $C\backslash C'$. Rather the definition requires only that all candidates in $C'$ are strictly preferred to those in $C\backslash C'$. Also notice that a set of voters $N'$ may be a solid coalition for multiple sets of candidates and that the entire set of voters $N$ is trivially a solid coalition for the set of all candidates $C$.} 



\bartonnew{
\begin{definition}[$q$-PSC]
Let $q\in (n/(k+1), n/k]$. We say a committee $W$ satisfies $q$-PSC if for every positive integer $\ell$, and for every solid coalition $N'$ supporting a candidate subset $C'$ with size $|N'|\ge \ell q$, the following holds
$$|W\cap C'|\ge \min\{\ell, |C'|\}.$$
\end{definition}
}

 If $q=n/k$, then we refer to the property as Hare-PSC. If $q=n/(k+1)+\epsilon$ for small $\epsilon>0$ , then we refer to the property as Droop-PSC.\footnote{Droop PSC is also referred to as Droop's proportionality criterion (DPC). Technically speaking the Droop quota is $\floor{n/(k+1)}+1$. The exact value $n/(k+1)$ is referred to as the  \emph{Hagenbach-Bischoff} quota.} 
 
There are some reasons to prefer the `Droop' quota ${n/(k+1)}+\epsilon$ for small $\epsilon>0$. Firstly, for $k=1$ the use of the Droop quota leads to rules that return a candidate that is most preferred by more than half of the voters. Secondly, STV defined with respect to the Droop quota ensures slight majorities get slight majority representation. 
Hare-PSC was stated as an essential property that a rule designed for PR should satisfy~\citep{Dumm84a}. 
When preferences are strict and $k=1$, \citet{Wood97a} refers to the restriction of Droop-PSC under these conditions as the \emph{majority principle}.  The majority principle requires that if a majority of voters are solidly committed to a set of candidates $C'$, then one of the candidates from $C'$ must be selected.


%
%

\begin{example}
									Consider the profile with 9 voters and where $k=3$. Then the voters in set $N'=\{1,2,3\}$ form a solid coalition with respect to Hare quota who support candidates in $\{c_1,c_2,c_3,c_4\}$. The voters in set $N''=\{4,5,6,7,8,9\}$ form a solid coalition with respect to Hare quota who support three candidate subsets $\{e_1\}$, $\{e_1, e_2\}$ and $\{e_1, e_2, e_3\}$. 
								\begin{align*}
														1:&\quad c_1, c_2, c_3, c_4,... \\
														2:&\quad c_4, c_1, c_2, c_3,...\\
														3:&\quad c_2, c_3, c_4, c_1,... \\
														4:&\quad e_1, e_2, e_3, ...\\
														5:&\quad e_1, e_2, e_3,  ...\\
														6:&\quad  e_1, e_2, e_3, ...\\
														7:&\quad e_1, e_2, e_3, ...\\
														8:&\quad  e_1, e_2, e_3, ...\\
														9:&\quad  e_2, e_1, e_3, ...
	\end{align*}
	\end{example}

One can also define a weak version of PSC. In some works~\citep[see, e.g.,][]{EFSS14a, EFSS17a,FSST17a}, the weaker version has been attributed to the original definition of PSC as defined by Dummett. For example, \citet{FSST17a} in their Definition 2.9 term \emph{weak} PSC as the property put forth by Dummett although he advocated the stronger property of PSC. 


\bartonnew{
\begin{definition}[weak $q$-PSC]
Let $q\in (n/(k+1), n/k]$. A committee $W$ satisfies weak $q$-PSC if for every positive integer $\ell$, and for every solid coalition $N'$ supporting a candidate subset $C'\, :\, |C'|\le \ell$ with size $|N'|\ge \ell q$, the following holds. 
$$|W\cap C'|\ge \min\{\ell, |C'|\}.$$
\end{definition}
}

For weak $q$-PSC, we restrict our attention to solid coalitions who support sets of candidates of size at most $\ell$ whereas in $q$-PSC we impose no such restriction. 
Note that $q$-PSC implies weak $q$-PSC but the reverse need not hold. Furthermore, the condition that $|C'|\le \ell$ and $|W\cap C'|\ge \min\{\ell, |C'|\}$ is equivalent to $C'\subseteq W$. We also note that under strict preferences and $k=1$, if a majority of the voters have the same most preferred candidate, then weak Droop-PSC implies that the candidate is selected. In particular, this implies that the majority principle is satisfied when weak-Droop PSC is satisfied.  

We now present a lemma connecting (weak) $q$-PSC for different values of $q$. The proof is omitted since it is implied by a stronger lemma (Lemma~\ref{Lemma: q and q' generalised PSC}) proven in section~\ref{Section: Generalised PSC axioms}. 

\begin{lemma}\label{Lemma: q and q' PSC}
Let $q, q'$ be real numbers such that $q<q'$.  If a committee $W$ satisfies (weak) $q$-PSC then $W$ satisfies (weak) $q'$-PSC.
\end{lemma}

%

%

\begin{remark}
{\color{black} In this paper we focus on PR axioms related to $q$-PSC where $q$ is a real number contained in the interval $(\frac{n}{k+1}, \frac{n}{k}]$. The reason for focusing on these values is that a committee satisfying $q$-PSC committee is guaranteed to exists for any preference profile when $q>\frac{n}{k+1}$. Whilst whenever $q\le \frac{n}{k}$ a \emph{unanimity}-like principle is satisfied. That is,  if all voters form a solid coalition for a size-$k$ candidate subset $C'$ then there is a unique committee satisfying $q$-PSC, i.e, $W=C'$. However in principle, the PR axioms and many results in this paper can be considered with values of $q$ outside of the interval $(\frac{n}{k+1}, \frac{n}{k}]$.}
\end{remark}

\bigskip

In section~\ref{Section: Generalised PSC axioms}, we generalise the PSC property to weak preferences which has not been done in the literature.


 \subsection{Candidate Monotoncity Axioms}

PR captures the requirement that cohesive groups of voters should get sufficient representation. Another desirable property is \emph{candidate monotonicity} that requires that \bartonnew{increased support for an otherwise-elected candidate should never cause this candidate to become unelected.} Candidate monotonicity involves the notion of a candidate being reinforced. 
\harisnew{We say a candidate is \emph{reinforced} if its relative position is improved while not changing the relative positions of all other candidates. }
\harisnew{More formally, we say that candidate $c$ is \emph{reinforced} in preference $\pref_i$ to obtain preference $\pref_i'$, if (1)
$c \pref_i d \implies c \pref_i' d$ for all $d\in C\setminus \{c\}$; (2)
$d \pref_i' c \implies d \pref_i c$ for all $d\in C\setminus \{c\}$; (3)
there exists a $d\in C$ such that $d \pref_i c$ and $c \succ_i' d$ (in this case $c$ is said to cross over $d$); and
(4) $d \pref_i e \implies d \pref_i' e$ for all $d,e \neq c$.
}
We are now in a position to formalise some natural candidate monotonicity properties of voting rules. The definitions apply not just to strict preferences but also to weak preferences. \harisnew{One of the definitions (RRCM) is based on the ranks of candidates as specified in the preliminaries.} 



\begin{definition}[Candidate Monotonicity]
	\noindent
 \begin{itemize}
 	\item \emph{Candidate Monotonicity (CM)}: if a winning candidate is reinforced {\color{black} by a single voter}, it remains a winning candidate.
 	\item \emph{Rank Respecting Candidate Monotonicity (RRCM)}: if a winning candidate $c$ is reinforced {\color{black} by a single voter} without changing the respective ranks of other winning candidates in each voter's preferences, then $c$ remains a winning candidate. 
 \item \emph{Non-Crossing Candidate Monotonicity (NCCM)}: if a winning candidate $c$ is reinforced {\color{black} by a single voter} without ever crossing over another winning candidate, then $c$ remains a winning candidate. 
 \item \emph{Weak Candidate Monotonicity (WCM)}: if a winning candidate is reinforced {\color{black} by a single voter}, then some winning candidate still remains winning. 
 \end{itemize}
 \end{definition}


NCCM and WCM are extremely weak properties but STV violates them even for $k=1$.  
We observe the following relations between the properties. 
 \begin{proposition}
 	The following relations hold. 
	
 	\begin{itemize}
\item  \harisnew{\textnormal{RRCM} and \textnormal{NCCM} are equivalent for linear orders.} 
\item $\textnormal{CM} \implies \textnormal{RRCM} \implies \textnormal{NCCM}$.
\item $\textnormal{CM} \implies \textnormal{WCM}$.
		\item Under $k=1$, \textnormal{WCM}, \textnormal{RRCM}, \textnormal{NCCM}, and \textnormal{CM} are equivalent. 
		\item Under dichotomous preferences, \textnormal{RRCM} and \textnormal{CM} are equivalent.
 	\end{itemize}
 	\end{proposition}
	
	Note that if a rule fails CM for $k=1$, then it also fails RRCM, NCCM, and WCM.

\section{PR under generalised preference relations}\label{Section: Generalised PSC axioms}

The notion of a solid coalition and PSC can be generalised to the case of weak preferences. In this section, we propose a new axiom called \emph{generalised PSC} which not only generalises PSC (that is only defined for strict preferences) but also \emph{Proportional Justified Representation} (PJR) a PR axiom that is only defined for dichotomous preferences.

\begin{definition}[Generalised solid coalition]
A set of voters $N'$ is a \emph{generalised solid coalition} for a set of candidates $C'$ if every voter in $N'$ weakly prefers every candidate in $C'$ at least as high as every candidate in $C\backslash C'$. That is, for all $i\in N'$ and for any $c'\in C'$
$$\forall c\in C\backslash C' \quad c'\pref_i c. $$
\end{definition}

We note that under strict preferences, a generalised solid coalition is equivalent to solid coalition. 
Let $c^{(i,j)}$ denotes voter $i$'s $j$-th most preferred candidate. In case the voter's preference has indifferences, we use lexicographic tie-breaking to identify the candidate in the $j$-th position.


\harisnew{
\begin{definition}[Generalised $q$-PSC]
Let $q\in (n/(k+1), n/k]$. A committee $W$ satisfies \emph{generalised $q$-PSC} if for all generalised solid coalitions $N'$ supporting candidate subset $C'$ with size $|N'|\ge \ell q$, there exists a set $C''\subseteq W$ with size at least $\min\{\ell, |C'|\}$ such that for all $c''\in C''$
$$\exists i\in N'\, : \quad c''\pref_i c^{(i, |C'|)}.$$
\end{definition}
}

\harisnew{The idea behind generalised $q$-PSC is identical to that of $q$-PSC and in fact generalised $q$-PSC is equivalent to $q$-PSC under linear preferences. 
Note that in the definition above, a voter $i$ in the solid coalition of voters $N'$ does not demand membership of candidates from the solidly supported subset $C'$ but of any candidate that is at least as preferred as a least preferred candidate in $C'$. Generalised weak $q$-PSC is a natural weakening of generalised $q$-PSC in which we require that $C'$ is of size at most $\ell$.}

\bartonnew{
\begin{definition}[Generalised weak $q$-PSC]
Let $q\in (n/(k+1), n/k]$. A committee $W$ satisfies weak \emph{generalised $q$-PSC} if for every positive integer $\ell$, and every generalised solid coalition $N'$ supporting a candidate subset $C': \, |C'|\le \ell$ with size $|N'|\ge \ell q$, there exists a set $C''\subseteq  W$ with size at least $\min\{\ell, |C'|\}$ such that for all $c''\in C''$
$$\exists i\in N'\, : \quad c''\pref_i c^{(i, |C'|)}.$$
\end{definition}
}



The following example shows that generalised $q$-PSC is a weak property when solid coalitions equal, or just barely exceed, the quota $q$. 

\begin{example}
Let $N=\{1, 2, 3, 4\}$, $C=\{a, b, \ldots, j\}$, $k=2$, and suppose voter $1$ and $2$'s preferences are given as follows:
\begin{align*}
1:&\qquad \{a,b,\ldots, j\}\\
2:&\qquad a, b,  \ldots , j
\end{align*}
We consider generalised PSC with respect to the Hare quota; that is, $q_H=n/k=2$. There is a generalised solid coalition $N'=\{1, 2\}$ with $|N'|\ge q_H$ supporting candidate subset $C'=\{a\}$. The generalised $q_H$-PSC axiom requires the election of $\ell=1$ candidates into $W$ who are at least as preferred as either voter $1$ or $2$'s most preferred candidate. Since voter $1$ is indifferent between all candidates, electing any candidate such as $j\in C$, will satisfy the axiom -- this is despite candidate $j$ being voter $2$'s strictly least preferred candidate.
\end{example}



We now show that if a committee $W$ satisfies generalised (weak) $q$-PSC then the committee also satisfies generalised (weak) $q'$-PSC for all $q'>q$.

\begin{lemma}\label{Lemma: q and q' generalised PSC}
Let $q, q'$ be real numbers such that $q<q'$.  If a committee $W$ satisfies generalised (weak) $q$-PSC then $W$ satisfies generalised (weak) $q'$-PSC.
\end{lemma}

\begin{proof}
Let $q<q'$ and suppose that the committee $W$ satisfies generalised (weak) $q$-PSC. We wish to show that (weak) generalised $q'$-PSC is also satisfied by $W$.  To see this, notice that any generalised solid coalition $N'$ requiring representation under generalised (weak) $q'$-PSC also requires at least as much representation under generalised (weak) $q$-PSC since $|N'|\ge \ell q'$ implies that $|N'|\ge \ell q$.
\end{proof}

\harisnew{Under linear orders, generalised PSC and generalised weak PSC coincide respectively with PSC and weak PSC.} 
Generalising PSC to the case of weak preferences is important because it provides a useful link with PR properties defined on dichotomous preferences. Proportional Justified Representation (PJR)~\citep{SFF+17a, AzHu16a} is a proportional representation property for dichotomous preferences~\citep{ABC+16a}.

Recall the following definition of PJR:

\begin{definition}[PJR]
A committee $W$ with $|W|=k$ satisfies PJR for an approval ballot profile $\boldsymbol{A}=(A_1, \ldots, A_n)$ over a candidate set $C$ if for every positive integer $\ell\le k$ there does not exists a set of voters $N^*\subseteq N$ with $|N^*|\ge \ell \frac{n}{k}$ such that 
$$\big|\bigcap_{i\in N^*} A_i\big|\ge \ell \qquad \text{but} \qquad \big|\big(\bigcup_{i\in N^*} A_i\big)\cap W\big|<\ell.$$
\end{definition}

\begin{proposition}\label{prop: Hare-PSC implies PJR}
Under dichotomous preferences, generalised weak Hare-PSC implies PJR. 
\end{proposition}

\begin{proof}
For the purpose of a contradiction, let $W$ be a committee of size $k$ and suppose that generalised weak Hare-PSC holds but PJR does not. If PJR does not hold, then there must exist a set $N^*$ of voters and a positive integer $\ell$ such that $|N^*|\ge \ell\frac{n}{k}=\ell q_H$ (where $q_H$ is the Hare quota) and both
	\begin{align}\label{HareImpliesPJR}
	\big|\bigcap_{i\in N^*} A_i\big|\ge \ell \qquad \text{and} \qquad \big|\big(\bigcup_{i\in N^*} A_i\big)\cap W\big|<\ell.
	\end{align}
	Note that if $i\in N^*$ it must be that $i$ is not indifferent between all candidates (i.e. $A_i\neq \emptyset, C$), otherwise (\ref{HareImpliesPJR}) cannot hold. \bartonnew{This means that the result is independent of whether a voter with preference $\pref_i$ leading to single equivalence class is defined to have preference presented via the approval ballot $A_i=\emptyset$ or $A_i=C$ (both of which induce the same single equivalence class over candidates). } 

Now it follows that $N^*$ is a generalised solid coalition for each candidate subset $C'\subseteq \bigcap_{i\in N^*} A_i$ since every candidate in $C'$ is weakly preferred to every candidate in $C$ for all $i\in N^*$. Since $\big|\bigcap_{i\in N^*} A_i\big|\ge \ell$, we can select a subset $C'$ with exactly $\ell$ candidates so that $|C'|=\ell$. 

	Thus, if generalised weak Hare-PSC holds then there exists a set $C''\subseteq W$ with size at least $ \min\{\ell, |C'|\}=\ell$ such that for all $c''\in C''$ there exists $i\in N^*\, :\, c''\pref c^{(i, \ell)}$. But note that for any voter $j\in N^*$ we have $c^{(j, \ell)}\in A_j$ and hence for this particular candidate $c''$ and voter $i\in N^*$ we have $c''\in A_i$. It follows that $C''\subseteq (\bigcup_{i\in N^*} A_i)\cap W$, and
	$$ \big|\big(\bigcup_{i\in N^*} A_i\big)\cap W\big|\ge |C''| \ge \ell,$$
		which contradicts (\ref{HareImpliesPJR}).
\end{proof}

\begin{proposition}
Under dichotomous preferences,  PJR implies generalised Hare-PSC. 
\end{proposition}
	\harisnew{
		\begin{proof}
		Suppose that for dichotomous preferences,
			a committee $W$ of size $k$ satisfies PJR. 
			Then there exists no set of voters $N^*\subseteq N$ with $|N^*|\ge \ell \frac{n}{k}$ such that 
			$$\big|\bigcap_{i\in N^*} A_i\big|\ge \ell \qquad \text{but} \qquad \big|\big(\bigcup_{i\in N^*} A_i\big)\cap W\big|<\ell.$$
	Equivalently, for every set of voters $N^*\subseteq N$ with $|N^*|\ge \ell \frac{n}{k}$, the following holds:
			$$\big|\bigcap_{i\in N^*} A_i\big|\ge \ell \implies \big|\big(\bigcup_{i\in N^*} A_i\big)\cap W\big|\geq \ell.$$
	We now prove that for all generalised solid coalitions $N^*$ of size $|N^*|\ge \ell n/k=\ell q_H$ (where $q_H$ is the Hare quota) supporting candidate subset $C'$ then there exists a set $C''\subseteq W$ with size at least $\min\{\ell, |C'|\}$ such that for all $c''\in C''$
		$$\exists i\in N^*\, : \quad c''\pref_i c^{(i, |C'|)}.$$
Consider a solid coalitions $N^*$ of size $|N^*|\ge \ell n/k$ supporting candidate subset $C'$.
	\begin{enumerate}
		\item Suppose there exists some voter $i\in N^*$ who has one of her least preferred candidates $c$ in $C'$. In that case, each candidate in $C'$ is at least as preferred for $i$ as $c$. Hence the condition of genaralized Hare-PSC is trivially satisfied by any committee. 
	\item The other case is that for each $i\in N^*$ and each $c\in C'$, $c\in \max_{\pref_i}(C)$.\footnote{Here $\max_{\pref_i}(C)$ denote the equivalence class of (strictly) most preferred candidates in $C$ with respect to $\pref_i$.} Equivalently, for each $i\in N^*$ and each $c\in C'$, $c\in A_i$. Hence $C'\subseteq \bigcap_{i\in N^*}A_i$.	Since $W$ satisfies PJR, it follows that $\big|\big(\bigcup_{i\in N^*} A_i\big)\cap W\big|\geq \ell.$ In that case, we know that there exists a set $C''=\big(\bigcup_{i\in N^*} A_i\big)\subseteq W$ of size at least $\min\{\ell, |C'|\}$ such that for all $c''\in W$, $$\exists i\in N^*\, : \quad c''\pref_i c^{(i, \ell)}.$$ Hence the condition of genaralised Hare-PSC is again satisfied. 
	\end{enumerate}
	This completes the proof. 
		\end{proof}
		}

%
%
%


\begin{corollary}
Under dichotomous preferences, PJR, weak generalised Hare-PSC, and generalised Hare-PSC are equivalent.
\end{corollary}

Since it is known that testing PJR is coNP-complete~\citep{AzHu16a,AEH+18}, it follows that testing generalised PSC and generalised weak PSC is coNP-complete. 

\begin{corollary}
Testing generalised PSC and generalised weak PSC is coNP-complete even under dichotomous preferences.  
\end{corollary}

On the other hand, PSC and weak PSC can be tested efficiently (please see the appendix).

Figure~\ref{fig:rel} depicts the relations between the different PR axioms.

							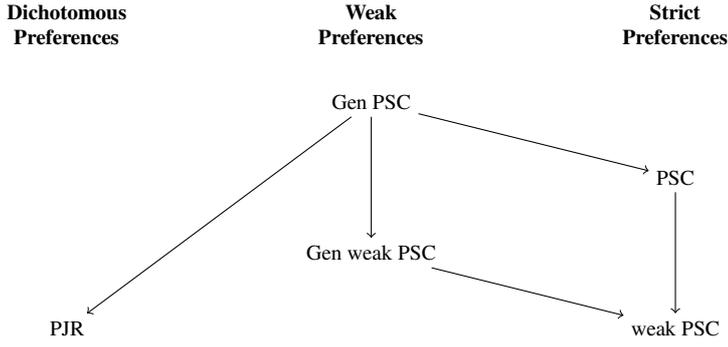
\begin{figure}[t]
								\begin{center}
						\begin{tikzpicture}
							\tikzstyle{pfeil}=[->,>=angle 60, shorten >=1pt,draw]
							\tikzstyle{onlytext}=[]

\node        (weak) at (0,10) {\begin{tabular}{c} \textbf{Weak}\\ \textbf{Preferences} \end{tabular}};

\node        (strict) at (4,10) {\begin{tabular}{c} \textbf{Strict}\\ \textbf{Preferences} \end{tabular}};

\node        (strict) at (-4,10) {\begin{tabular}{c} \textbf{Dichotomous}\\ \textbf{Preferences} \end{tabular}};


			\node        (G-PSC) at (0,9) {Gen PSC};
			
						\node        (G-w-PSC) at (0,7) {Gen weak PSC};


			\node        (PSC) at (4,8) {PSC};

			\node        (PJR) at (-4,6) {PJR};
			
						\node        (w-PSC) at (4,6) {weak PSC};			
						
	\draw[->] (G-PSC) -- (G-w-PSC);
	
	\draw[->] (G-PSC) -- (PJR);

	\draw[->] (PSC) -- (w-PSC);
	
	\draw[->] (G-PSC) -- (PSC);
	\draw[->] (G-w-PSC) -- (w-PSC);

						\end{tikzpicture}
						\end{center}
						\caption{Relations between properties. 
						 An arrow from (A) to (B) denotes that (A) implies (B). For any concept A, A with respect to Droop quota is stronger than A with respect to Hare quota. }
						 \label{fig:rel}
						\end{figure}

%
%
%
%
%
%
%
%
%
%
%
%
%

\section{The Case of STV}

In this section, we define the family of STV rules for instances where voters submit strict preferences. The family is formalised as Algorithm~\ref{algo:stv}. 
STV is a multi-round rule; in each round either a candidate is selected as a winner or one candidate is eliminated from the set of potential winners. Depending on the quota $q$ and the reweighting rule applied, one can obtain particular STV rules~\citep[see e.g., ][]{AlKa13a}. {\color{black} To distinguish between variants of STV, based on different quota values $q$, we denote the STV rule with quota $q$ by $q$-STV. When the quota $q$ is equal to the Hare or Droop quota we simply refer to the $q$-STV variant as Hare-STV or Droop-STV, respectively.} One of the most common rules is Hare-STV with discrete reweighting. This implies that a subset of voters of size {\color{black} $\ceil{n/k}$} 
 is removed from the profile once their most preferred candidate in the current profile has been selected. STV modifies the preference profile $\succ$ by deleting candidates. We will denote by $C(\succ)$ the current set of candidates in the profile $\succ$. When $k=1$, STV is referred to as \emph{Instant-Runoff voting (IRV)} or as the \emph{Alternative Vote (AV)}.

														\begin{algorithm}[h!]
															  \caption{STV family of Rules}
															  \label{algo:stv}
							\normalsize
															\begin{algorithmic}
																
																	\REQUIRE  $(N,C,\succ, k)$, quota $q\in (\frac{n}{k+1},\frac{n}{k}]$. \COMMENT{$\succ$ is a profile of strict preferences}

																\ENSURE $W\subseteq C$ such that $|W|=k$
															\end{algorithmic}
															\begin{algorithmic}[1]
																\normalsize
		

			 						\STATE $W\longleftarrow \emptyset$; 
										\STATE $w_i\longleftarrow 1$ for each $i\in N$							
				
										\STATE $j\longleftarrow 1 $														\WHILE{$|W|<k$}
								\IF{$|W|+ |C(\succ)|=k$}
								\RETURN $W\cup C(\succ)$
								\ENDIF						\IF{there is a candidate $c$ with plurality support (i.e., the total weight of voters who have $c$ as the most preferred candidate among candidates in $C(\succ)$)
								at least $q$ }
							\STATE \label{reweighting line STV} Let the set of voters supporting $c$ be denoted by $N'$. Modify the weights of voters in $N'$ so the total weight of voters in $N'$ decreases by $q$.  
						\STATE Remove $c$ from the profile $\succ$.						
							\STATE $W\longleftarrow W\cup\{c\}$
							\ELSE
							\STATE Remove a candidate with the lowest plurality support from the current preference profile $\succ$
							\ENDIF						
								\ENDWHILE
								\RETURN $W$
								\end{algorithmic}
								\end{algorithm}

%


								
								\begin{example}\label{Example: Droop-STV}[Illustration of Droop-STV]
							{	For this example we consider the Droop-STV rule with \emph{uniform fractional reweighting}. This reweighting method means that line~\ref{reweighting line STV} of Algorithm~\ref{algo:stv} is executed as follows: First calculate the total weight of voters in $N'$, i.e.,  $T=\sum_{i\in N'} w_i$, then the weight of each voter $i\in N'$ is updated from $w_i$ to $w_i\times \frac{T-q_D}{T}$ where $q_D$ is the Droop quota. }
								
								To illustrate this STV rule consider the following profile with 9 voters and suppose we wish to elect a committee of size $k=3$ 
								
												\begin{align*}
													1:&\quad c_1, c_2, c_3, e_1, e_2, e_3, e_4, d_1\\
													2:&\quad c_2, c_3, c_1, e_1, e_2, e_3, e_4, d_1\\
													3:&\quad c_3, c_1, d_1, c_2, e_1, e_2, e_3, e_4\\
													4:&\quad e_1, e_2, e_3, e_4, c_1, c_2, c_3, d_1\\
													5:&\quad e_1, e_2, e_3, e_4, c_1, c_2, c_3, d_1\\
													6:&\quad  e_1, e_2, e_3, e_4, c_1, c_2, c_3, d_1\\
													7:&\quad e_1, e_2, e_3, e_4, c_1, c_2, c_3, d_1\\
													8:&\quad  e_1, e_2, e_3, e_4, c_1, c_2, c_3, d_1\\
													9:&\quad  e_1, e_2, e_3, e_4, c_1, c_2, c_3, d_1
\end{align*}
					
													In the first round $e_1$ is selected and the total weight of the voters in set $\{4, 5, 6, 7, 8, 9\}$ goes down by the Droop quota $q_D$, i.e., slightly more than $2.25$. Candidate $e_1$ is then removed from the preference profile. In the second round,  $e_2$ is selected and the total weight of the voters in the set $\{4, 5, 6, 7, 8, 9\}$ is now $6-2q_D$, i.e., slightly less than 1.5. Candidate $e_2$ is then removed from the preference profile. After that since no candidate has plurality support, with respect to current weights, of at least the quota one candidate is deleted. Candidates $e_3$ and $e_4$ are removed in succession as they have plurality support no more than $1.5$ with respect to the current voting weights. Then candidate $c_1$ is elected since she has plurality support of $1+6-2q_D$, i.e., slightly less than 2.5, which exceeds the quota $q_D$.\end{example}

STV has been claimed to satisfy Proportionality for Solid Coalitions/Droop Proportionality Criterion~\citep{Dumm84a, Wood94a}. 
On the other hand, STV violates just about every natural monotonicity axiom that has been proposed in the literature.  

In STV, voters are viewed as having an initial weight of one. When a candidate supported by a voter is selected, the voter's weight may decrease. 
STV can use fractional reweighting or discrete reweighting.\footnote{Fractional reweighting in STV has been referred to as Gregory  or `senatorial'~\citep[see e.g., ][]{Jans16a, Tide95a}}  We will show that fractional reweighting is crucial for some semblance of PR. Incidentally, fractional reweighting is not necessarily introduced to achieve better PR but primarily to minimize the ``stochastic aspect'' of tie-breaking in STV~\citep[pp 32, ][]{Tide95a}. The following result shows that if STV resorts to discrete reweighting then it does not even satisfy weak PSC. Discrete reweighting refers to the modification of voter weights in Line 9 of Algorithm~\ref{algo:stv} such that the total weight of voters in $N'$ decreases by some integer greater or equal to $q$.


%

\begin{proposition}
Under strict preferences and for any $q, q'\in (\frac{n}{k+1}, \frac{n}{k}]$, $q$-STV with discrete reweighting does not satisfy weak $q'$-PSC. 
	\end{proposition}
	\begin{proof}
	
		Let $N=\{1,2, \ldots, 10\}$, $C=\{c_1, \ldots, c_8\}$, $k=7$ and consider the following profile:
		\begin{align*}
			1:&\quad c_1, c_5, c_6, c_7, c_8, c_2, c_3, c_4\\
			2:&\quad c_2, c_5, c_6, c_7, c_8, c_1, c_3, c_4\\
			3:&\quad c_3, c_5, c_6, c_7, c_8, c_1, c_2, c_4\\
			4:&\quad c_4, c_5, c_6, c_7, c_8, c_1, c_2, c_3\\
			5:&\quad c_5, c_6, c_7, c_8, c_1, c_2, c_3, c_4\\
			6:&\quad c_5, c_6, c_7, c_8, c_1, c_2, c_3, c_4\\
			7:&\quad c_5, c_6, c_7, c_8, c_1, c_2, c_3, c_4\\
			8:&\quad c_5, c_6, c_7, c_8, c_1, c_2, c_3, c_4\\
			9:&\quad c_5, c_6, c_7, c_8, c_1, c_2, c_3, c_4\\
			10:&\quad c_5, c_6, c_7, c_8, c_1, c_2, c_3, c_4
			\end{align*}
			
For fixed $q\in  (\frac{n}{k+1}, \frac{n}{k}]$ we use $q$-STV with discrete reweighting to select the candidates. Under discrete reweighting, the total weights of voters are modified by some integer $p\ge 2$ (refer to Line 9 of Algorithm~\ref{algo:stv}). In this proof we focus on the case where $p=2$, a similar argument can be applied to prove the proposition for larger integer values.  Note that $2=\ceil{q}$ for any $q\in  (\frac{n}{k+1}, \frac{n}{k}]$.
			
			 Applying the $q$-STV rule, first $c_5, c_6$ and $c_7$ are selected. Each time we select these candidates, the total weight of voters in the set $N'=\{5, 6, 7, 8, 9, 10\}$ goes down by 2. Thus the remaining four candidates are to be selected from $c_1, c_2, c_3, c_4$ and $c_8$. At this stage candidate $c_8$ has the lowest plurality support (equal to zero) and is removed from all preference profiles and the list of potentially elected candidates, and hence $c_8\notin W$.
			
			Now for any fixed $q'\in(\frac{n}{k+1}, \frac{n}{k}]$, weak $q'$-PSC requires that at least four candidates from $\{c_5,c_6, c_7,c_8\}$ be selected, since $|N'|\ge 4\times q'$, but using discrete reweighting only three candidates are selected by $q$-STV. 
				\end{proof}
				
				The proof above has a similar argument as Example 1 in \citep{SEL17a} that concerns an approval voting setting.

		Next we show that STV satisfies PSC with fractional reweighting. The proof details are in the appendix.

	\begin{proposition}\label{prop:stv-psc-satisfy}
	For any $q\in (\frac{n}{k+1}, \frac{n}{k}]$, under strict preferences $q$-STV satisfies $q$-PSC.	\end{proposition}

		Below we provide an example of STV violating CM for $k=1$. { The example assumes the Droop quota, i.e., Droop-STV, however, the same example violates CM for any $q$-STV rule with $q\in (\frac{n}{k+1}, \frac{n}{k}]$.}
		
		\begin{example}\label{Example: STV violation of CM}[Example showing that STV violates CM for $k=1$.]
									Consider the following instance of 100 voters with strict preferences.
					\begin{align*}
					\text{Total number of voters }& \quad \text{Corresponding preferences}\\
					 28 :&\quad c,b,a\\
					5:& \quad c,a,b\\
					30 :&\quad  a,b,c\\
					5:&\quad  a, c,b\\
					16:&\quad  b,c,a\\
					16:&\quad b,a,c.
					\end{align*}
					 We consider the single-winner election setting with the Droop quota; that is, $k=1$ and $q_D=50+\varepsilon$ for sufficiently small $\varepsilon>0$.
 
					 Under the Droop-STV rule the outcome is $W_{STV}=\{a\}$. To see this notice that the plurality support of the candidates $a,b,c$ are $35,32,33$, respectively. Since no candidate receives plurality support $\ge q_D$ we remove the candidate with lowest plurality support, i.e, candidate $b$, and the 32 voters previously supporting candidate $b$ now give their plurality support to their second preference. Thus, the updated plurality support of the two remaining candidates $a$ and $c$ are $35+16=51$ and $33+16=49$, and hence candidate $a$ is elected.

					 Now to show a violation of CM we consider an instance where two voters originally with preferences $c,a,b$ change their preferences to $a,c,b$  i.e. a reinforcement of the previously winning candidate $a$. The new profile is shown below
					 \begin{align*}
										\text{Total number of voters }& \quad \text{Corresponding preferences}\\
					28 \spref:&\quad c,b,a\\
					3\spref:& \quad c,a,b\\
					30\spref:&\quad  a,b,c\\
					7\spref:&\quad  a, c,b\\
					16\spref:&\quad  b,c,a\\
					16\spref:&\quad b,a,c.
					\end{align*}
					In this modified setting, the Droop-STV outcome is $W_{STV}'=\{b\}$ which is a violation of candidate monotonicity (CM). To see this notice that plurality support of the candidates $a,b,c$ are $37, 32, 31$, respectively. Since no candidate receives plurality support $\ge q_D$ we remove the candidate with lowest plurality support, i.e, candidate $c$, and the 31 voters previously supporting candidate $c$ now give their plurality support to their second preference. Thus, the updated plurality support of the two remaining candidates $a$ and $b$ are $37+3=40$ and $32+28=60$, and hence candidate $b$ is elected.
					
									\end{example}

\section{Expanding Approvals Rule (\newrule)}

We now present the Expanding Approvals Rule (\newrule). The rule utilises the idea of $j$-approval voting whereby every voter is asked to approve their $j$ most preferred candidates, for some positive integer $j$. At a high level, \newrule works as follows. 

\begin{quote}
An index $j$ is initialised to $1$. The voting weight of each voter is initially 1. We use a quota $q$ that is between $n/(k+1)$ and $n/k$.
While $k$ candidates have not been selected, we do the following. 
We perform $j$-approval voting with respect to the voters' current voting weights. If there exists a candidate $c$ with approval support at least a quota $q$, we select such a candidate. 
If there exists no such candidate, we increment $j$ by one and repeat until $k$ candidates have been selected.  
\end{quote}

							\begin{algorithm}[h!]
								  \caption{Expanding Approvals Rule (\newrule)}
								  \label{algo:newrule}
\normalsize
								\begin{algorithmic}
									\REQUIRE  $(N,C,\pref, k)$ parametrised by quota $q\in (\frac{n}{k+1},\frac{n}{k}]$. \COMMENT{$\pref$ can contain weak orders; if a voter $i$ expresses her preferences over a subset $C'\subset C$, then $C\setminus C'$ is considered the last equivalence class of the voter.}
									\ENSURE $W\subseteq C$ such that $|W|=k$
								\end{algorithmic}
								\begin{algorithmic}[1]
									\normalsize
		\STATE \label{step:priority} Use some default priority ordering $L$ over $C$.
			\STATE $w_i\longleftarrow 1$ for each $i\in N$							
				
			\STATE $j\longleftarrow 1 $														\WHILE{$|W|<k$}

		\WHILE{there does not exist a candidate in $C\setminus W$ with support at least $q$ in a $j$-approval vote}
		\STATE $j\longleftarrow j+1$				
									\ENDWHILE
									\STATE \label{step-j-approval} Among the candidates with support at least $q$ in a $j$-approval election, select the candidate $c$ from $C\backslash W$ that has highest ranking wrt $L$. \COMMENT{Voters are asked to approve their $j$ most preferred candidates and any candidates that are at least as preferred as the $j$--th most preferred candidate.}
						
									\STATE \label{step-reweighting} 
									
			Let the set of voters supporting $c$ in the $j$-approval election be denoted by $N'$. Modify the weights of voters in $N'$ so the total weight of voters in $N'$ decreases by exactly $q$. 						
									
		%
	\ENDWHILE
								
									\RETURN $W$
								\end{algorithmic}
							\end{algorithm}

The rule is formally specified as Algorithm~\ref{algo:newrule}. 
It is well-defined for weak preferences. 
EAR is based on a combination of several natural ideas that have been used in the design of voting rules.

\begin{enumerate}
	\item Candidates are selected in a sequential manner. 
	\item A candidate needs to have at least the Droop quota of `support' to be selected.  
	\item The voting weight of a voter is reduced if some of her voting weight has already been used to select some candidate. The way voting weight is reduced is fractional.
\item 

We use $j$-approval voting for varying $j$. When considering weak orders, we adapt $j$-approval voting so that in $j$-approval voting, a voter not only approves her $j$ most preferred candidates but also any candidate that is at least as preferred as the $j$-th most preferred candidate.  \harisnew{One way to view $j$-approval for weak orders is as follows. 
(1) break all ties temporarily to get an artificial linear order. (2) Identify the $j$-th candidate $d$ in the artificial linear order. (3) Approve all candidates that are at least as weakly preferred as $d$. }
\end{enumerate}


For EAR, the default value of $q$ that we propose is 
\[\bar{q}=\frac{n}{k+1} + \frac{1}{m+1} \left(\floor{\frac{n}{k+1}}+1-\frac{n}{k+1}\right).\]
The reason for choosing this quota is that
		$\bar{q}$ can be viewed as $\bar{q}=\frac{n}{k+1}+\varepsilon$ where $\varepsilon$ is small enough so that
		that for any $\ell\leq k$, $\ell\cdot \bar{q} <\ell\frac{n}{k+1}+1$. {\color{black} In particular, this implies that if there exists a solid coalition $N'$ of size $|N'|\ge \ell q_D$ (where $q_D$ is the Droop quota) then $|N'|\ge \ell \bar{q}$.} On the other hand, $\varepsilon$ is large enough so that the algorithm is polynomial in the input size.
	
	We also propose a default priority ordering that is with respect to rank maximality under $\pref$. This way of tie breaking is one but not the only way to ensure that  \newrule satisfies RRCM. For any candidate $a$, its \emph{corresponding rank vector} is $r(a)=(r_1(a), \ldots, r_m(a))$ where $r_j(a)$ is the number of voters who have $a$ in her $j$-th most preferred equivalence class.
We compare rank vectors lexicographically. One rank vector $r=(r_1,\ldots, r_m)$ is \emph{better} than $r'=(r_1',\ldots, r_m')$ if for the smallest $i$ such that $r_i\neq r_i'$, it must hold that $r_i>r_i'$.

Finally we propose the following natural way of implementing the reweighting in Step~\ref{step-reweighting}.
If the total support for $c$ in the $j$-approval election is $T$, then for each $i\in N$ who supported $c$, we reweigh it as follows
 \[w_i\longleftarrow w_i\times \frac{T-q}{T}.\]
This ensures that exactly $q$ weight is reduced.

	%

In the following example, we demonstrate how \newrule works. 



\begin{example}\label{Example: Illustration new rule}[Illustration of \newrule]
Consider the profile with 9 voters and where $k=3$. Note that the default quota is $\bar{q}\approx 2.33$.
				\begin{align*}
					1:&\quad c_1, c_2, c_3, e_1, e_2, e_3, e_4, d_1\\
					2:&\quad c_2, c_3, c_1, e_1, e_2, e_3, e_4, d_1\\
					3:&\quad c_3, c_1, d_1, c_2, e_1, e_2, e_3, e_4\\
					4:&\quad e_1, e_2, e_3, e_4, c_1, c_2, c_3, d_1\\
					5:&\quad e_1, e_2, e_3, e_4, c_1, c_2, c_3, d_1\\
					6:&\quad  e_1, e_2, e_3, e_4, c_1, c_2, c_3, d_1\\
					7:&\quad e_1, e_2, e_3, e_4, c_1, c_2, c_3, d_1\\
					8:&\quad  e_1, e_2, e_3, e_4, c_1, c_2, c_3, d_1\\
					9:&\quad  e_1, e_2, e_3, e_4, c_1, c_2, c_3, d_1
					\end{align*}
					
					In the first round $e_1$ is selected and the total weight of the voters in set $\{4, 5, 6, 7, 8, 9\}$ goes down by $\bar{q}$, i.e., $6-\bar{q}\approx 3.67$. In the second round, since no other candidate has sufficient weight when we run the 1-approval election, we consider 2-approval election. Under 2-approval, candidate $e_2$ receives support $6-\bar{q}\approx 3.67$ which exceeds $\bar{q}$ and hence is elected.  When $e_2$ is selected, the total weight of the voters in set $\{4, 5, 6, 7, 8, 9\}$  goes down again by $\bar{q}$. At this point the total weight of voters in set $\{4, 5, 6, 7, 8, 9\}$ is $6-2\bar{q}\approx 1.34$ and no unselected candidate has approval support more than $\bar{q}$, under 2-approval. So \newrule considers 3-approval whereby, candidates $c_1$ and $c_3$ both get support 3. Hence \newrule selects $c_1$, due to the default priority ordering $L$ (rank maximality ordering), and the winning committee is $\{e_1,e_2, c_1\}.$
					
					\end{example}
		The example above produced an \newrule outcome which coincided with the Droop-STV outcome (recall Example~\ref{Example: Droop-STV}). Next, we present an example showing that STV and EAR are different even for $k=1$. 
		

									\begin{example}[Example showing that STV and EAR are different even for $k=1$.]
					Recall Example~\ref{Example: STV violation of CM} which showed that the following instance of 100 voters with strict preferences
					\begin{align*}
					\text{Total number of voters }& \quad \text{Corresponding preferences}\\
28 :&\quad c,b,a\\
					3:& \quad c,a,b\\
					30:&\quad  a,b,c\\
					7:&\quad  a, c,b\\
					16:&\quad  b,c,a\\
					16:&\quad b,a,c,
					\end{align*}
					 leads to the single-winner Droop-STV outcome $W=\{b\}$. { More generally for any quota $q\in (\frac{n}{k+1}, \frac{n}{k}]$ the $q$-STV outcome is also $\{b\}$.} We now show that in this case the \newrule outcome is $W'=\{a\}$ and hence the STV outcome and \newrule outcomes do not coincide even for $k=1$.
					 
					 To see that the \newrule outcome is $W'=\{a\}$ we proceed as follows: First, we consider the 1-approval election which gives candidates $a,b,c$ approval support $37,32,31$, respectively. Since no candidate attains approval support beyond the default quota $\bar{q}=50\frac{1}{3}$ we move to the 2-approval election. In the 2-approval election candidates $a,b,c$ attain approval support $56,90,54$, respectively. Since all candidates attain approval support beyond $\bar{q}$ we apply our default priority ordering $L$ (rank-maximality) which leads to candidate $a$ being elected.

 
 
									\end{example}

					\begin{remark}\label{Remark: tweaks to EAR}
						In case voters do not specify certain candidates in their list and do not wish that their vote weight to be used to approve such candidates, EAR can be suitably tweaked so as to allow this requirement. In this case, candidates are selected as long as a selected candidate can get approval weight $q$. The required number of remaining candidates can be selected according to some criterion. 
Another way EAR can be varied is that instead of using $L$ as the priority ordering, the candidate with the highest weighted support that is at least $q$ is selected.
						\end{remark}

We point out the rule's outcome can be computed efficiently.

		\begin{proposition}
			\newrule runs in polynomial time $O(n+m)^{2}$.
			\end{proposition}
			
			
			The rank maximal vectors can be computed in $O(n+m)$ and the ordering based on rank maximality can be computed in time $O(n+m)^{2}$. In each round, the smallest $j$ is found for which there are some candidates in $C\setminus W$ that have an approval score of at least $q$. The candidate $c$ which is rank-maximal is identified. All voters who approved of $c$ have their weight modified accordingly which takes linear time. Hence the whole algorithm takes time at most $O(n+m)^{2}$.

	A possible criticism of \newrule is that the choice of quota as well as reweighting makes it complicated enough to not be usable by hand or to be easily understood by the general public. However we have shown that without resorting to fractional reweighting, STV already fails weak PSC. 

				Since \newrule is designed for proportional representation which is only meaningful for large enough $k$, \newrule may not be the ideal rule for $k=1$. Having said that, we mention the following connection with a single-winner rule from the literature. 
				
				
				\begin{remark}
					For $k=1$ and under linear orders with every candidate in the list, \newrule is equivalent to the Bucklin voting rule~\citep{BrSa09a}. For $k=1$ and under linear orders for all but a subset of equally least preferred candidates applying the tweak in Remark~\ref{Remark: tweaks to EAR} leads to the \newrule being equivalent to the Fallback voting rule~\citep{BrSa09a}. 
					\end{remark}
%

%

Under dichotomous preferences and using Hare quota, \newrule bears similarity to Phragm{\'e}n's first method (also called Enstr{\"{o}}m's method) described by \citet{Jans16a} (page 59).
However the latter method when extended to strict preferences does not satisfy Hare-PSC. Although \newrule has connections with previous rules, extending them to the case of multiple-winners and to handle dichotomous, strict and weak preferences simultaneously and satisfy desirable PR properties requires careful thought. 

%


We observe some simple properties of the rule. It is anonymous (the names of the voters do not matter). It is also neutral as long as lexicographic tie-breaking is not required to be used. 
Under linear orders and when using \newrule with the default quota, if more than half the voters most prefer a candidate, then that candidate is selected. This is known as the majority principle.

{ \newrule is defined with default quota $q$. However it is possible to consider variants of \newrule with other quota values such as the Hare quota. We refer to the variant of \newrule with the Hare quota as Hare-\newrule. The proposition below shows that the choice of this quota can lead to different outcomes. In particular, we show that \newrule (defined with the default quota) can lead to a different outcome than that attained under Hare-\newrule.}


		\begin{proposition}\label{prop:hare-droop}
		Under dichotomous preferences, Hare-EAR and EAR are not equivalent. 
		\end{proposition}				 

		\begin{proof}
		Let $|N|=100$ and $k=5$. Denote the Hare quota by $q_H=20$ and note that the default quota is $\bar{q}\approx 16.71$. Let the preferences of the voters be (i.e. expressing the top equivalence class of most preferred candidates):
		\begin{align*}
\text{Total number of voters }& \quad \text{Corresponding preferences}\\
		17: &\qquad \{a\}\\
		  17: &\qquad \{b,f\}\\
		 17: &\qquad \{c\}\\
		 17: &\qquad \{d\}\\
		 17: &\qquad \{e\}\\
		  1:&\qquad \{f\}\\
		  14: &\qquad \{g\}.
		\end{align*}

		First note that the rank-maximal ordering, $L$ (with the lexicographic ordering applied for equal rank maximal ordering) is 
		$$f\triangleright a\triangleright b\triangleright c\triangleright d\triangleright e\triangleright g.$$
		This ordering is independent of voter weights, is calculated at the start of the algorithm, and is never updated throughout the algorithm. 

		Now under Hare-EAR, in the 1-approval election no candidate has support exceeding $q_H$. Thus, we move to a 2-approval election whereby all voters support all candidates and all candidates have support $100$. Thus, we select the candidate with highest rank-maximal priority which is candidate $f$ and reweight its supporters (i.e. reduce all voter weights to $8/10$ since all voters support all candidates). Repeating this process leads to the election of the first five candidates according to the priority ranking $L$ i.e.
		$$W=\{f,a,b,c,d\}.$$ 

		Now under EAR, in the 1-approval election there are 6 candidates $\{a,b,c,d,e, f\}$ with support $\bar{q}$. Thus, we select the candidate with highest L priority, i.e., candidate $f$, and then reweight each of its 18 supporters' weights to $\frac{18-\bar{q}}{18}\approx 0.07$. Now there still remains 4 candidates $\{a,c,d,e\}$ with support of $17>\bar{q}$, since all voters supporting such a candidate are disjoint we have that the winning committee is 
		$$W'=\{f,a,c,d,e\}.$$

		Note that $W\neq W'$.
		\end{proof}

\section{Proportional Representation and Candidate Monotonicity under \newrule}

\subsection{Proportional Representation  under \newrule}

We show that \newrule satisfies the central PR axioms for general weak order preference profiles.

							\begin{proposition}
								Under weak orders, \newrule satisfies generalised Droop-PSC.
								\end{proposition}

		\begin{proof}
		Let $W$ be an outcome of the \newrule and suppose for the purpose of a contradiction that generalised Droop-PSC is not satisfied. That is, there exists a positive integer $\ell$ and a generalised solid coalition $N'$ such that $|N'|\ge \ell q_D$ (where $q_D$ is the Droop quota) supporting a candidate subset $C'$ and for every set $C''\subseteq W$ with $|C''|\ge \min\{\ell, |C'|\}$ there exists $c''\in C''$ such that 
		\begin{align}\label{strictcontracondition}
		\forall i\in N' \,\qquad c^{(i, |C'|)}\spref_i c''.
		\end{align}
		Without loss of generality we may assume that the solid coalition $N'$ is chosen such that $\ell \le |C'|$ and so $\min\{\ell, |C'|\}=\ell$.

		Let $j^*$ be the smallest integer such that in the $j^*$-approval election each voter in $N'$ supports all candidates in $C'$. We claim that the $j^*$-approval election must be reached. Suppose not, then it must be that $|W|=k$ at some earlier $j$-approval election where $j<j^*$. But if $|W|=k$ this implies that after reweighting
			\begin{align}\label{sumofweights2}
			\sum_{i\in N} w_i&= n-k\bar{q}=\frac{n}{k+1}-k\bar{\varepsilon},
			\end{align}
			where $\bar{\varepsilon}=\frac{1}{m+1}\Big(\floor{\frac{n}{k+1}}+1-\frac{n}{k_1}\Big)$. However, in every $j'$-approval election for $j'\le j$ each voter $i\in N'$ only supports candidates weakly preferred to $c^{(i, |C'|)}$. Let $\hat{C}_i$  be the set of candidates which voter $i$ finds weakly preferable to $c^{(i, |C'|)}$ and define $\hat{C}=\cup_{i\in N'} \hat{C}_i$. The total weight of voters in $N'$ at the termination of the algorithm (i.e. at the end of the $j$-approval election) is reduced by at most $|C''| \bar{q} $ where $C''=W\cap \hat{C}$. But by assumption since $C'' \subseteq W$ and every candidate $c''\in C''$ is weakly preferred to $c^{(i, |C'|)}$ for some voter $i\in N'$ it must be that $|C''|<\min\{\ell, |C'|\}=\ell$. It follows that
			\begin{align}\label{sumofweights3}
			\sum_{i\in N'} w_i\ge \ell q_D-(\ell-1)\bar{q}= \frac{n}{k+1}-(\ell-1)\bar{\varepsilon},
			\end{align}
			which contradicts (\ref{sumofweights2}) since $N'\subseteq N$. Thus, we conclude that the $j^*$-approval election is indeed reached.	

		Now at the $j^*$-approval election each voter $i\in N'$ supports only the candidates in the set $\hat{C}_i$, excluding those already elected in an earlier approval election. Recall that $\hat{C}=\cup_{i\in N'} \hat{C}_i$ and let $m^*$ be the number of candidates elected from $\hat{C}$ in earlier approval elections. Since $|N'|\ge \ell q_D$ implies that $|N'|\ge \ell \bar{q}$, it follows that the total weight of voters in $N'$ when the $j^*$-approval election is reached is at least $(\ell-m^*)\bar{q}$. Since $C'\subseteq \hat{C}_i$ for all $i\in N'$, every unelected candidate in $C'$ attains support at least $(\ell-m^*)\bar{q}$, note that there are at least $|C'|-m^*\ge \ell-m^*$ of these. In addition, each voter $i\in N'$ also supports the unelected candidates in $\hat{C}_i\backslash C'$. Thus, the \newrule algorithm can only terminate if weight of voters in $N'$ is reduced below $\bar{q}$ which can only occur if at least $(\ell-m^*)$ candidates from $\hat{C}$ are elected. It then follows that 
		$$|W\cap \hat{C}|=m^*+(\ell-m^*)=\ell.$$
	By again defining $C''=W\cap \hat{C}$ we attain a contradiction since $C''\subseteq W$ and $|C''|\ge \ell$ but (\ref{strictcontracondition}) does not hold since $C''\subseteq \hat{C}$.
%
%
%
%
			\end{proof}


		\begin{remark}
			Note that \newrule satisfying PSC or generalised PSC does not depend on  what priority tie-breaking is used (Step~\ref{step:priority}) or how the fractional reweighting is applied (Step~\ref{step-reweighting}). 
			\end{remark}
%
%

		
										Recalling Lemma~\ref{Lemma: q and q' generalised PSC} we have the following corollary.
	
									\begin{corollary}\label{Corollary: EAR satisfies gen Hare PSC}
										Under weak orders, \newrule satisfies generalised Hare-PSC.
										\end{corollary}
									
														%
					
								\begin{corollary}
								Under linear orders, \newrule satisfies the majority principle. 
								\end{corollary}

								\begin{proof}
									Under linear orders, Droop-PSC implies the majority principle. 
									\end{proof}

			We get the following corollary from Corollary~\ref{Corollary: EAR satisfies gen Hare PSC}.
			
			\begin{corollary}\label{cor:pjr}
			Under dichotomous preferences, \newrule satisfies proportional justified representation.  
				\end{corollary}
				\begin{proof}
					We had observed that under dichotomous preferences, generalised Hare-PSC implies PJR. 
					\end{proof}
			
				Incidentally, the fact that there exists a polynomial-time algorithm to satisfy PJR for dichotomous preferences was the central result of two recent papers~\citep{SFF16a,BFJL16a}. We have shown that EAR can in fact satisfy a property stronger than PJR that is defined with respect to the Droop quota.
					Since \newrule satisfies generalised PSC, it implies that there exists a polynomial-time algorithm to compute a committee satisfying generalised PSC. Interestingly, we already observed that checking whether a given committee satisfies generalised PSC is coNP-complete. 	
				
		 %

						%
						%
						

\subsection{Candidate Monotonicity under \newrule}

%

We show that  \newrule satisfies rank respecting candidate monotonicity (RRCM). In what follows we shall refer to the profile of all voter preferences (weak or strict) as simply the \emph{profile}.

							\begin{proposition}
								\newrule satisfies rank respecting candidate monotonicity (RRCM).
			\end{proposition}
			
			
			\begin{proof}
				Consider a profile $\pref$ with election outcome $W$ and let $c_i\in W$. Now consider another modified profile $\pref'$ in which $c_i$'s rank is improved while not harming the rank of other winning candidates, relative to $\pref$, and denote the election outcome under $\pref'$ by $W'$. Since we use rank-maximality to define the order $L$, note that the relative position of $c_i$ is at least as good under $\pref'$ as it is under $\pref$.
				
				Let the order of candidates selected under $\pref$ be $c_1,\ldots, c_i, \dots, c_{|W|}$. In the modified profile $\pref'$, let us trace the order of candidates selected. For the first candidate $c_1$, either it is selected first for exactly the same reason as it is selected first under $\pref$ or, alternatively, now $c_i$ is selected. If $c_i$ is selected, our claim has been proved. Otherwise, the same argument is used for candidates after $c_1$ until $c_i$ is selected.
\end{proof}

	This leads immediately to the following corollaries of the above proposition.

 				\begin{corollary}
 					\newrule satisfies  non-crossing candidate monotonicity (NCCM).
 					\end{corollary}

				\begin{corollary}
					For $k=1$, \newrule satisfies candidate monotonicity (CM). 
					\end{corollary}
					
					\begin{corollary}
										For dichotomous preferences, \newrule satisfies candidate monotonicity.
									\end{corollary}
									\begin{proof}
		Consider dichotomous profile $P$ and another profile dichotomous $P'$ in which winning candidate $c_i$'s rank is improved while not harming the rank of other winning candidates. This implies that the rank of $c_i$ is improved while not affecting the ranks of any other alternatives including the winning candidates. Since \newrule satisfies rank respecting candidate monotonicity (RRCM), it follows that for dichotomous preferences, \newrule satisfies candidate monotonicity. 					
										\end{proof}

%
%

	On the other hand, \newrule does not satisfy CM or WCM for $k>1$.	
						
	\begin{proposition}
	\newrule does not satisfy WCM. 
\end{proposition}

	\begin{proof}
	Let $N=\{1, 2, 3, 4\}$, $C=\{a,b, c,d,e, f\}$, $k=2$, and let strict preferences be given by the following preference profile:
\begin{align*}
1:&\quad a, c, f, d,\ldots \\
2:&\quad d, b, f, a, c, \ldots \\
3:&\quad a, d, c, b, f, \ldots \\
4:&\quad f, e, c, d, \ldots,
\end{align*}
	With this preference profile \newrule will output the winning committee 
	$W=\{a, f\}$. First, $a$ is selected into $W$ and each weight of each voter $i\in\{1,3\}$ is reduced to $w_i=\frac{2-\bar{q}}{2}$ where $\bar{q}$ is the default quota. Since $\bar{q}>\frac{4}{3}$ we infer that $w_i<1/3$.  Moving to the $2$-approval election no candidate receives support of at least $\bar{q}$. Finally in the $3$-approval election candidate $f$ receives support at least $2$ which exceeds the quota, and candidate $c$ attains support $1+2w_i$. If $1+2w_i<\bar{q}$ then candidate $f$ is the only candidate exceeding the quota and so is elected. On the other hand, if $1+2w_i\ge \bar{q}$ then both candidates $c$ and $f$ exceed the quota; however, due to the default (rank maximal) priority ordering $f$ is still elected into $W$.

	\bigskip
	
	Now consider a reinforcement of $f$ by voter $1$ (shift $f$ from third to first place), this is described by the following preference profile:
\begin{align*}
1:&\quad f, a, c, d, \ldots \\
2:&\quad d, b, f, a, c,\ldots\\
3:&\quad  a, d, c, b, f,\ldots\\
4:&\quad c, e, f, d,\ldots,
\end{align*}
	With these preferences the winning committee is $W'=\{d, c\}$. In the 2-approval election both candidates $a$ and $d$ attain support of at least $\bar{q}$, however due to the priority ordering $L$ (rank maximality) candidate $d$ is selected into $W$ and each voter $i\in \{2,3\}$ has their weight reduced to $w_i=\frac{2-\bar{q}}{2}<\frac{1}{3}$. In the 3-approval election both candidates $c$ and $f$ attain support of at least $2\ge \bar{q}$. Furthermore, $c$ and $f$ are equally ranked with respect to rank maximality but applying lexicographic tie-breaking leads to $c$ being elected into $W'$.
\end{proof}

\begin{remark}
	The above proposition was proven for the default quota $\bar{q}$ however the same counter-example can be used to prove the statement for any other quota $q\in (\frac{n}{k+1}, \frac{n}{k}]$.
	\end{remark}

			\section{Other Rules}	
			\label{sec:otherrules}

			In the literature, several rules have been defined for PR purposes. We explain how \newrule is better in its role at achieving a strong degree of PR or has other relative merits.
			
			\subsection{Quota Borda System (QBS)}

			\citet{Dumm84a} proposed a counterpart to STV called QPS (Quota Preference Score)  or a more specific version QBS (Quota Borda System). The rule works for complete linear orders and is designed to obtain a committee that satisfies Droop-PSC. It does so by examining the prefixes (of increasing sizes) of the preference lists of voters and checking whether there exists a corresponding solid coalition for a set of voters. If there is such a solid set of voters, then the appropriate number of candidates with the highest Borda count are selected.\footnote{The description of the rule is somewhat informal and long in the original books of Dummett which may have lead to The Telegraph terming the rule  as a ``a highly complex arrangement''~\citep{Tele11a}. }
We can partition the PSC demands among demands pertaining to rank from $j=1$ to $m$. For any given $j$, we only consider candidates in $C^{(j)}$
the set of candidate involved in preferences of voters up till their first $j$ positions. In Algorithm~\ref{algo:QBS}, we present a formal description of QBS.

							\begin{algorithm}[h!]
								  \caption{Quota Borda System (QBS)}
								  \label{algo:QBS}
\normalsize
								\begin{algorithmic}
									\REQUIRE  $(N,C,\succ, k)$ 
									 \COMMENT{$\succ$ is a profile of strict preferences} 							\ENSURE $W\subseteq C$ such that $|W|=k$
								\end{algorithmic}
								\begin{algorithmic}[1]
									\normalsize
		 
		\STATE Set  quota $q$ as some value $>n/(k+1)$ 
								\STATE $W\longleftarrow \emptyset$

			\STATE $j\longleftarrow 1 $					
			\WHILE{$j<m$}
				\WHILE{there does not exist a solidly supported set of candidates in $C^j$ whose demand is not met by $W$}
		\STATE $j\longleftarrow j+1$	
									\ENDWHILE
						
	\STATE Partition the set of voters into equivalence classes where each class has the same solidly supported set of $j$ candidates. \COMMENT{The next while loop in called the \emph{selection while loop} in which the selection of candidates in a stage occurs.}				

	\WHILE{there exists an equivalence class of voters $N'\subseteq N$ whose PSC demand with respect to first $j$ candidates is not met by $W$}

	\STATE Among the candidates in the first $j$ positions that are solidly supported by $N'$, select a candidate  $c\notin W$ that has the highest Borda score. 	
	\STATE $W\longleftarrow W\cup \{c\}$
	\ENDWHILE
	\ENDWHILE

								
									\RETURN $W$
								\end{algorithmic}
							\end{algorithm}

Although Dummett did not show that the rule satisfies some axiom which STV does not, he argued that QPS satisfies the Droop proportionality criterion and is somewhat less ``chaotic'' than STV. \citet{Schu02a} argues that QBS is chaotic as well and his Example 3 implicitly shows that QBS in fact violates WCM.\footnote{\citet{Gell02a} wrongly claims that QBS satisfies the stronger axiom of CM.} 
Tideman also feels that QBS is overly designed to satisfy Droop-PSC but is not robust enough to go beyond this criterion especially if voters in a solid coalition perturb their preferences. 
			
\newrule has some important advantages over QBS: (1) it can easily handle indifferences whereas QBS is not well-defined for indifferences. In particular, in order for QBS to be suitably generalised for indifferences and to still satisfy generalised PSC, it may become an exponential-time rule\footnote{QBS checks for PSC requirements and adds suitable number of candidates to represent the corresponding solid coalition of voters. In order to work for generalised PSC, QBS will have to identify whether solid coalitions of voters and meet their requirement which means that it will need to solve the problem of testing generalised-PSC which is coNP-complete.}, (2)  \newrule can easily handle voters expressing partial lists by implicitly having a last indifference class whereas QBS cannot achieve this, (3) \newrule rule satisfies an established PR property called PJR for the case for dichotomous preferences. As said earlier, QBS is not well-defined for indifferences and even for dichotomous preferences, (4) \newrule addresses a criticism of \citet{Tide06a}: ``Suppose there are voters who would be members of a solid coalition except that they included an ``extraneous'' candidate, which is quickly eliminated  among their top choices. These voters' nearly solid support for the coalition counts for nothing which seems to me inappropriate.'' We demonstrate the last flaw of QBS pointed out by Tideman in the explicit example below. \newrule does not have this flaw.

\begin{example}

Consider the profile with 9 voters and where $k=3$.
				\begin{align*}
					1:&\quad c_1, c_2, c_3, e_1, e_2, e_3, e_4, d_1\\
					2:&\quad c_2, c_3, c_1, e_1, e_2, e_3, e_4, d_1\\
					3:&\quad c_3, c_1, d_1, c_2, e_1, e_2, e_3, e_4\\
					4:&\quad e_1, e_2, e_3, e_4, c_1, c_2, c_3, d_1\\
					5:&\quad e_1, e_2, e_3, e_4, c_1, c_2, c_3, d_1\\
					6:&\quad  e_1, e_2, e_3, e_4, c_1, c_2, c_3, d_1\\
					7:&\quad e_1, e_2, e_3, e_4, c_1, c_2, c_3, d_1\\
					8:&\quad  e_1, e_2, e_3, e_4, c_1, c_2, c_3, d_1\\
					9:&\quad  e_1, e_2, e_3, e_4, c_1, c_2, c_3, d_1
					\end{align*}
					In the  example, $\{e_1, e_2, e_3\}$ is the outcome of QBS. Although PSC is not violated for voters in $\{1,2,3\}$ but the outcome appears to be unfair to them because they almost have a solid coalition. Since they form one-third of the electorate they may feel that they deserve that at least one candidate such as $c_1$, $c_2$ or $c_3$ should be selected. {\color{black} In contrast, it was shown in Example~\ref{Example: Illustration new rule} that \newrule does not have this flaw and instead produces the outcome $\{e_1, e_2, c_1\}$.}
					\end{example}

\subsection{Chamberlin-Courant and Monroe}

			There are other rules that have been proposed within the class of ``fully proportional representation'' rules such as Monroe~\citep{Monr95a} and Chamberlin-Courant (CC)~\citep{ChCo83a}. 
\citet{Monr95a} used the term ``fully proportional representation'' to refer to PR-oriented rules that take into account the full preference list. 
			Recently, variants of the rules called Greedy Monroe, and Greedy CC~\citep{EFSS14a, EFSS17a} have been discussed. However none of the rules satisfy even weak PSC~\citep{EFSS14a, EFSS17a}. 
			A reason for this is that voters are assumed to not care about how many of their highly preferred candidates are in the committee as long the most preferred is present. Monroe and CC are also NP-hard to compute~\citep{PSZ08a}. 
			
			
	\subsection{Phragm{\'e}n's First Method}

	Phragm{\'e}n's first method was first considered by Phragm{\'e}n but not published or pursued by him~\citep{Jans16a}. In the method, voters approve of most preferred candidates that have not yet been selected. The candidate with highest weight of approval is selected. The total weight of the voters whose approved candidate was selected is reduced by the Hare quota if the total weight is more than the Hare quota. Otherwise all such voters' weights are set to zero. 
	
	Although Phragm{\'e}n's first method seems to have been defined primarily for dichotomous preferences, the same definition of the rule works for linear orders. However the rule when applied to linear orders does not satisfy Hare-PSC. See the example below. 
	%
	%
	%
	%
	%
	%
	%
								
									\begin{example}

									Consider the profile with 9 voters and where $k=3$.
													\begin{align*}
														1:&\quad c_1, c_2, c_3, \ldots\\
														2:&\quad c_2, c_3, c_1, \ldots\\
														3:&\quad c_3, c_1, c_3,\ldots\\
														4:&\quad e_1, c_1, c_2, c_3, \ldots\\
														5:&\quad e_2, c_1, c_2, c_3,\ldots\\
														6:&\quad e_3, c_1, c_2, c_3,\ldots\\
														7:&\quad e_4, c_1, c_2, c_3,\ldots\\
														8:&\quad  e_5, c_1, c_2, c_3,\ldots\\
														9:&\quad  e_6, c_1, c_2, c_3,\ldots
		\end{align*}
					
														In the  example, $\{e_1, e_2, e_3\}$ is a possible outcome of rule. When $e_1$ is selected, voter $4$'s weight goes to zero. Then when $e_2$ is selected, voter $5$'s weight goes to zero. Finally $e_3$ is selected. Hare-PSC requires that $c_1$ or $c_2$, or $c_3$ is selected. 
														\end{example}
								
	%
	
	The next example shows that Hare-EAR and Phragm{\'e}n's first rule are not equivalent. 
								\begin{example}
								Consider the profile with 4 voters, where $k=2$ and with dichotomous preferences given as follows:
									\begin{align*}
														1:&\quad \{b\} \\
														2:&\quad \{a, c\} \\
														3,4:&\quad \{a\}.
		\end{align*}
		Note that in this example the Hare quota is $q_H=2$. Now under Hare-\newrule the winning committee is $W=\{a,c \}$. First $a$ is elected since her approval support is $3\ge q_H$ and all other candidates have support less than the quota. Once candidate $a$ is elected all of her supporting voters have their weights reduced to $w_i=\frac{1}{3}$. Now there is no candidate with support beyond the quota, $q_H$, and so we move to a 2-approval election. In this case the support of candidate $b$ is $2-\frac{1}{3}<q_H$ and the support of candidate $c$ is $2\ge q_H$, and so $c$ is elected into $W$.  
		
		 Under Phragm{\'e}n's first rule we attain $W'=\{a,b\}$. First candidate $a$ is elected since she has maximal support, then all supporting voters have their weight reduced to $\frac{1}{3}$. Then the weighted supported of candidate $b$ is $1$ and the weight support of candidate $c$ is $\frac{1}{3}$ -- hence Phragm{\'e}n's first rule elects $b$ into the committee.

		 Thus, Hare-\newrule and Phragm{\'e}n's rule are not equivalent under dichotomous preferences. 

								\end{example}

\bigskip

\subsection{Phragm{\'e}n's  Methods}

			A compelling rule is  Phragm{\'e}n's Ordered Method that can even be generalised for weak orders. Under strict preferences, it satisfies weak Droop-PSC~~\citep[Theorem 16.1 (ii), ][]{Jans16a}. On the other hand, even under strict preferences, it does not satisfy Droop-PSC~\citep[page 51, ][]{Jans16a}. If we are willing to forego PSC, then Phragm{\'e}n's Ordered Method seems to be an exceptionally  useful rule for strict preferences because unlike STV it satisfies both candidate monotonicity and committee monotonicity~\citep{Jans16a}. Committee monotonicity requires that for any outcome $W$ of size $k$, there is a possible outcome $W'$ of size $k+1$ such that $W'\supset W$.

\subsection{Thiele's Ordered Methods}

			Thiele's methods are based on identifying candidates that are most preferred by the largest weight of voters. For any voter who has had $j$ candidates selected has current weight $1/(j+1)$. 
\citet{Jans16a} presented an example (Example 13.15) that can be used to show that Thiele's ordered methods do not satisfy weak Droop-PSC under strict preferences. 			
			

			\bigskip

			\subsection{CPO-STV rules}
			
			A class of STV related rules is CPO-STV that was proposed by \citet{Tide06a}. The rules try to achieve a PR-type objective while ensuring that for $k=1$, a Condorcet winner is returned if there exists a Condorcet winner.\footnote{Since these rules are Condorcet-consistent, they are vulnerable to the no-show paradox~\citep{Moul88}.}
One particular rule within this class is Schulze-STV~\citep{Schu11a}. All of these rules are only defined for linear orders and hence do not satisfy generalised PSC. Furthermore, they all require enumeration of all possible committees and then finding pairwise comparisons between them. Hence they are exponential-time rules and impractical for large elections. \citet{Tide95a} writes that CPO-STV is ``computationally tedious, and for an election with several winners and many candidates it may not be feasible.''
\citet{Tide06a} also considered whether CPO-STV rules satisfy Droop-PSC but was unable to prove that they satisfy PSC (page 282). In any case, having an exponential-time rule satisfying PSC may not be compelling because there exists a trivial exponential-time algorithm that satisfies PSC: enumerate committees, check whether they satisfy PSC, and then return one of them.

 
%
%
%
%
%


				\section{Conclusions}

				\begin{table}[h!]
	
				\begin{center}
				\begin{tabular}{llll}
				\toprule
				&STV&QBS&\newrule\\
				\midrule
				Generalised D-PSC / H-PSC&no&no&yes\\
								Generalised weak D-PSC / H-PSC&no&no&yes\\
					PJR&no&no&yes\\
				D-PSC /  H-PSC&yes&yes&yes\\
			
				Weak D-PSC /  H-PSC&yes&yes&yes\\
				
				\midrule
				CM&no&no&no\\
				CM for dichotomous preferences&no&no&yes\\
				CM for $k=1$&no&yes&yes\\
				RRCM&no&no&yes\\
				NCCM&no&yes&yes\\
				\midrule
				polynomial-time&yes&yes&yes\\
				\bottomrule
				\end{tabular}
				\end{center}
				\caption{Properties satisfied by STV, QBS, and \newrule. QBS and STV are the only rules from the literature that satisfy PSC and are computable in polynomial time. STV does not satisfy weak PSC if discrete reweighting is used.
				}
				\label{table:rules}
				\end{table}

				\begin{table}[h!]
	
				\begin{center}
				\begin{tabular}{lcc}
				\toprule
				&Complexity of&Complexity of\\
				&{Computing}&{Testing}\\
				\midrule
				PSC&in P&in P\\
				Weak PSC&in P&in P\\
				Generalised PSC&in P&coNP-complete\\
				Generalised weak PSC&in P&coNP-complete\\
				\bottomrule
				\end{tabular}
				\end{center}
				\caption{Computational complexity of {computing } a committee satisfying a PR property and {testing} whether a given committee satisfies a property.    
				}
				\label{table:complexity}
				\end{table}

			%
			%
			%

%
%

				In this paper, we undertook a formal study of proportional representation under weak preferences. The generalised PSC axiom we proposed generalises several well-studied PR axioms in the literature. We then devised a rule that satisfies the axiom. 
Since EAR has relative merits over STV and Dummett's QBS (two known rules that satisfy PSC), it appears to be a compelling solution for achieving PR via voting. At the very least, it appears to be another useful option in the toolbox of representative voting rules and deserves further consideration and study. The relative merits of STV, QBS, and \newrule are summarised in Table~\ref{table:rules}.

	EAR can be modified to also work for `participatory budgeting' settings in which candidates may have a non-unit `cost' and the goal is to select a maximal set of candidates such that the total cost of candidates does not exceed a certain budget $B$. In our original setting of committee voting, EAR with $q=n/k$ can be seen as follows. 
We initialise the weight of each voter as $1$. Each candidate is viewed as having a unit cost and the budget is $k$. A candidate $c$ is selected if including $c$ does not exceed the budget ($k$) and if $c$ has approval support of $1n/k$. When $c$ is selected, the total weight of the voters supporting $c$ is decreased by $n/k$ times the unit cost. If no such candidate exists, we increment $j$-approval to $j+1$-approval. The process continues until no candidate can be included without exceeding the budget limit.
We now explain how to modify EAR to work for non-unit costs and budgets.
We again initialise the budget weight of each voter as $1$. 
A candidate $c$ is selected if including $c$ does not exceed the budget $B$ and 
if $c$ has approval support of $n/B$ times the cost of $c$. When $c$ is selected, the total weight of the voters supporting $c$ is decreased by $n/B$ times the cost of $c$. If no such candidate exists, we increment $j$-approval to $j+1$-approval. The process continues until no candidate can be included without exceeding the budget limit. Proportional representation axioms related to those studied in this paper have only recently been extended and explored in the participatory budgeting setting by~\cite{ALT18}.

Our work also sheds light on the complexity of computing committees that satisfy PR axioms as well the the complexity of testing whether a given a committee satisfies a given PR axiom. We found that whereas a polynomial-time algorithm such as \newrule finds a committee that satisfies generalised PSC, testing whether a given committee satisfies properties such as generalised PSC or generalised weak PSC is computationally hard. These findings are summarised in Table~\ref{table:complexity}.

				\section*{Acknowledgments}
				Haris Aziz is supported by a Julius Career Award.
				Barton Lee is supported by a Scientia PhD fellowship. The authors thank Markus Brill, Edith Elkind, Dominik Peters and Bill Zwicker for comments. 

\bibliographystyle{plainnat}
%


 \appendix

\section{STV}

	\paragraph{Proof of Proposition~\ref{prop:stv-psc-satisfy}}

						\begin{proof}
						Let $q\in (\frac{n}{k+1}, \frac{n}{k}]$ and assume that all voters have strict preferences. Suppose for the purpose of a contradiction that $W$ is the STV outcome for some instance and $q$-PSC is not satisfied. That is, there exists a positive integer $\ell$ and a solid coalition $N'\subseteq N$ with $|N'|\ge \ell q$ supporting a set of candidates $C'$ and $W\cap C'|<\min\{\ell^*, |C'|\}$. Without loss of generality we assume that $|C'|\ge \ell$. If this were not the case, i.e., $|C'|<\ell$, one could simply consider a smaller subset of voters in $N'$ and these voters will necessarily still solidly support all candidates in $C'$.
						
						Thus, given that $N'$ solidly supports the candidate set $C'$ with $|C'|\ge \ell$ and $|N'|\ge \ell q$ we wish to derive a contradiction from the fact that 
						$$j=|W\cap C'|<\min\{\ell, |C'|\} =\ell.$$
						The STV algorithm iteratively elects or eliminates a single candidate.\footnote{Technically speaking, if the sum of elected candidates, $W$, and unelected but also uneliminated candidates, $C^*$, is equal to $k$ the set of all unelected and also uneliminated candidates, $C^*$, are elected simultaneously in the same iteration. Within this proof and for simplicity of notation, we assume that if such situation has occurred the election of the candidates in $C^*$ occurs in a sequential manner and supporting voter weights are reduced as per line 8-10 of Algorithm~\ref{algo:stv}. Clearly this assumption is without loss of generality.} Let $T$ be the number of the iteration which the STV algorithm runs through to output $W$, and let $(c_1, \ldots, c_T)$ be a sequence of candidates such that in iteration $t\in \{1, \ldots, T\}$ candidate $c_t$ is either elected or eliminated. Note that $T\ge k$ and it need not be the case that $\cup_t\{c_t\}=C$. Furthermore, since $|W\cap C'|=j$ it must be the case that $|\cup_t\{c_t\}\cap C'|\ge j$.

						First, we claim that $C'\subseteq \cup_t \{c_t\}$. If this were not the case, then in every iteration voters in $N'$ support only candidates in $C'$ and since precisely $j$ candidates in $C'$ were elected the remaining weight of voters in $N'$ is $>(\ell -j)q$. But this is a contradiction; the STV algorithm cannot terminate when voters in $N$ have remaining weight $>q$ (recall the footnote of the previous paragraph) since this would imply $<k$ candidates have been elected. Thus, we conclude that $C'\subseteq \cup_t \{c_t\}$. This implies that there exist a unique ordering of the candidates in $C'$
						$$(c_1', \ldots, c_{|C'|}')$$
						 which maintains the ordering $(c_1, \ldots, c_T)$, i.e., simply define the one-to-one mapping $\sigma: \{1, \ldots, |C'|\}\rightarrow \{1, \ldots, T\}$ such that $(c_1', \ldots, c_{|C'|}')=(c_{\sigma(1)}, \ldots, c_{\sigma(|C'|)})$. 
						 
						 Second, we claim that the $j$ candidates $(c_{|C'|-(j-1)}', c_{|C'|-(j-2)}', \ldots, c_{|C'|}')$ are all elected. Starting with iteration $t=\sigma(|C'|)$ suppose that candidate $c_{|C'|}'$ is not elected. This would then imply that at the start of iteration $t$ voters in $N'$ have weight $>(\ell-j)q>q$ -- having only supporting candidates in $C'$ in earlier iterations at precisely $j$ of these being elected. Furthermore, since  $c_{|C'|}'$ is the only unelected and un-eliminated candidate in $C'$ all voters in $N'$ support candidate $c_{|C'|}'$. Thus, the weighted-plurality support is $>q$ which contradicts the assumption that $c_{|C'|}'$ is not elected. We conclude that in fact candidate $c_{|C'|}'$ is elected and so in all earlier iterations $s<t$ the weight of voters in $>(\ell-j+1)q$ and precisely $(j-1)$ candidates from $C'$ are elected in earlier iterations. Now consider iteration $t'=\sigma\big(|C'|=1\big)$ and suppose that candidate $c_{|C'|-1}'$ is not elected. But at the start of this iteration each voter in $N'$ supports with their plurality vote either candidate $c_{|C'|}$ or candidate $c_{|C'|-1}$ and in total have weight $>(\ell-j+1)q>2q$. Thus, at least one of the candidate must have weighted-plurality support $>q$ and so no candidate is eliminated in this iteration, and we conclude that it must be that candidate $c_{|C'|-1}'$ is elected. This argument can be repeated for all iterations $t''=\sigma(|C'|-2), \ldots, \sigma\big(|C'|-(j-1)\big)$, and so we conclude that the $j$ candidates $(c_{|C'|-(j-1)}', c_{|C'|-(j-2)}', \ldots, c_{|C'|}')$ are all elected.
						 
						 Finally, we argue that candidate $c_{|C'|-j}$ is also elected which contradicts the original assumption that $|W\cap C'|=j<\ell$. To see this, note that at the start of iteration $t^*=\sigma(|C'|-j)$ the total weight of voters in $N'$ is $>\ell q$ (applying the claim from the previous paragraph) and each voter in $N'$ has a plurality vote for one of the $(j+1)$ candidates in $\cup_{i=0}^{j} \{c_{|C'|-i}'\}$. Noting that $\ell \ge j+1$, there exists at least one candidate in $\cup_{i=0}^{j} \{c_{|C'|-i}'\}$ who receives weighted-plurality support $> q$ and so no candidate is eliminated in this iteration. In particular, this implies that candidate $c_{|C'|-j}$ is elected. Combining this with the previous claims we see that $j+1$ candidate are elected which contradicts the original assumption and completes the proof.
						\end{proof}

 \section{Complexity of Testing PSC}
 
 \begin{proposition}
	 Under linear orders, it can be tested in polynomial time whether a committee satisfies PSC.
	 \end{proposition}
	 \begin{proof}
		 For each $i$ from $1$ to $m$ one can look at prefixes of preference lists of sizes $i$. For these prefixes, we can see check whether there exists a corresponding solid coalition. For such solid coalitions we can check whether the appropriate number of candidates are selected or not.
		 \end{proof}

	 The same idea can be used for weak PSC. 

	 \begin{proposition}
		 Under linear orders, it can be tested in polynomial time whether a committee satisfies weak PSC.
		 \end{proposition}

		 \begin{proposition}
			 Under dichotomous preference, the problem of testing generalised PSC is coNP-complete.	 
			 \end{proposition}
			 

			 \begin{proof}
				 Under dichotomous preferences, generalised PSC is equivalent to PJR. Since \citet{AzHu16a} and \citet{AEH+18} showed that testing PJR is coNP-complete, it follows that that testing Generalised PSC is coNP-complete as well. 
%
%
%
%
%
%
%
				 \end{proof}

%
%
%

\end{document}